\documentclass{llncs}
\usepackage{ifpdf}
\ifpdf
 \usepackage{pdfsync}
\fi

\usepackage{amsmath,amssymb,enumerate,gastex,wasysym,a4wide}

\newcommand{\bN}{\mathbb N}
\newcommand{\bQ}{\mathbb Q}
\newcommand{\cA}{\mathcal A}
\newcommand{\cB}{\mathcal B}
\newcommand{\cI}{\mathcal I}
\newcommand{\cL}{\mathcal L}
\newcommand{\cS}{\mathcal S}
\newcommand{\cM}{\mathcal M}
\newcommand{\CTreg}{\mathrm{CT_{reg}}}
\newcommand{\dom}{\mathrm{dom}}

\newcommand{\lex}{\mathrm{lex}}
\newcommand{\llex}{\mathrm{llex}}

\newcommand{\pref}{\mathrm{pref}}
\newcommand{\rest}{\mathord\restriction}
\newcommand{\Run}{\mathrm{Run}}
\newcommand{\trees}{\mathrm{trees}}
\newcommand{\VD}{\mathrm{VD}}
\newcommand{\wt}{\mathrm{wt}}
\newcommand{\word}{\mathrm{word}}

\title{Isomorphisms of scattered automatic linear orders}

\author{Dietrich Kuske} 

\institute{Institut f\"ur Theoretische Informatik, Technische
  Universit\"at Ilmenau, Germany}

\begin{document}
\maketitle

\begin{abstract}
  We prove that the isomorphism of scattered tree automatic linear
  orders as well as the existence of automorphisms of scattered word
  automatic linear orders are undecidable. For the existence of
  automatic automorphisms of word automatic linear orders, we
  determine the exact level of undecidability in the arithmetical
  hierarchy.
\end{abstract}

\section{Introduction}

Automatic structures form a class of computable structures with much
better algorithmic properties: while, due to Rice's theorem, nothing is
decidable about a computable structure (given as a tuple of Turing
machines), validity of first-order sentences is decidable in automatic
structures (given as a tuple of finite automata). This property of
automatic structures was first observed and exploited in concrete
settings by B\"uchi, by Elgot~\cite{Elg61}, and by Epstein et
al.~\cite{EpsCHLPT92}. Hodgson~\cite{Hod82} attempted a uniform
treatment, but the systematic study really started with the work by
Khoussainov and Nerode~\cite{KhoN95} and by Blumensath and
Gr\"adel~\cite{BluG00,BluG04}. Over the last decade, a fair amount of
results have been obtained, see e.g.\ the surveys \cite{Rub08,BarGR11}
as well as the list of open questions~\cite{KhoN08}, for very recent
results not covered by the mentioned articles, see e.g.\
\cite{BraS11,DurH12,JaiKSS12,Hus12}.

A rather basic question about two automatic structures is whether they
are isomorphic. For ordinals and Boolean algebras, this problem was
shown to be decidable via a characterisation of the automatic members
of these classes of structures. On the other hand, already Blumensath
and Gr\"adel~\cite{BluG04} observed that this problem is undecidable
in general. In \cite{KhoNRS07}, it is shown that the isomorphism
problem is $\Sigma^1_1$-complete; a direct interpretation yields the
same result for successor trees, for undirected graphs, for
commutative monoids, for partial orders (of height 2), for lattices
(of height 4)~\cite{Nie07}. Rubin \cite{Rub04} shows that the
isomorphism problem for locally finite graphs is complete
for~$\Pi^0_3$. In \cite{KusLL11}, we show in particular that also the
isomorphism problems of order trees and of linear orders are
$\Sigma^1_1$-complete. For the handling of linear orders, our
arguments rely heavily on ``shuffle sums''. Consequently, we construct
linear orders that contain a copy of the rational line (a linear order
not containing the rational line is called scattered, i.e., our result
is show for non-scattered linear orders). This is unavoidable since we
also show that the isomorphism problem for scattered linear orders is
reducible to true arithmetic (i.e., the first-order theory of
$(\bN,+,\cdot)$) and therefore much ``simpler'' than the isomorphism
problem for arbitrary linear orders. But it is still conceivable that
the isomorphism problem for scattered linear orders is decidable.

In this paper, we deal with automatic scattered linear orders. In
particular, we prove the following three results:
\begin{enumerate}[(1)]
\item There is a scattered linear order whose set of tree-automatic
  presentations is $\Pi^0_1$-hard (i.e.\ one can reduce the complement
  of the halting problem to this problem). This holds even if we fix
  the order relation on the set of all trees
  (Theorem~\ref{T-Iso-tree-automatic}). Hence also the isomorphism
  problem for tree automatic scattered linear orders is $\Pi^0_1$-hard
  (Corollary~\ref{C-Iso-tree-automatic}).
\item The existence of a non-trivial automorphism of an automatic
  scattered linear order is $\Sigma^0_1$-hard (i.e.\ the halting
  problem reduces to this problem,
  Corollary~\ref{C-rigid-automatic}). Again, this holds even if we fix
  the linear order on the set of all words
  (Theorem~\ref{T-rigid-automatic}). The existence of an automatic
  non-trivial automorphism is $\Sigma^0_1$-complete.

  For regular languages ordered lexicographically, the existence of a
  non-trivial automorphism is decidable
  (Theorem~\ref{T-regular-rigid}), but it becomes undecidable for
  contextfree languages (Theorem~\ref{T-rigid-cf}).
\item The existence of a non-trivial automorphism of a tree automatic
  scattered linear order is $\Sigma^0_2$-hard (i.e., one can reduce
  the set of Turing machines that accept a finite language to this
  problem, Theorem~\ref{T-rigid-tree-automatic}).
\end{enumerate}
The proof of (2) uses an encoding of polynomials similarly
to~\cite{KusLL11} but avoids the use of shuffle sums. The technique
for proving (1) and (3) is genuinely new: One can understand a
weighted automaton over the semiring $(\bN\cup\{-\infty\};\max,+)$ as
a classical automaton with a partition of the set of transitions into
two sets~$T_0$ and~$T_1$. The behavior of such a weighted automaton
assigns numbers to words~$w$, namely the maximal number of transitions
from $T_1$ in an accepting run on the word $w$. Krob \cite{Kro94}
showed that the equivalence problem for such weighted automata is
$\Pi^0_1$-complete. The hardness results from (1) are based on a
sharpening of Krob's result (see \cite{DroK}): there is a fixed
weighted automata such that the set of equivalent weighted automata is
$\Pi^0_1$-hard (and therefore undecidable). A closer analysis of this
proof, together with the techniques for proving (1) and (2), finally
yields~(3).

These results show that the existence of isomorphisms and of
automorphisms is nontrivial for scattered linear orders that are
described by word and tree automata, resp.

\section{Preliminaries}

\subsection{Tree and word automatic structures}

Let $\Sigma$ be some alphabet.  A \emph{$\Sigma$-tree} or just a
\emph{tree} is a partial mapping $t:\{0,1\}^*\to\Sigma$ such that
$uv\in\dom(t)$ implies $u\in\dom(t)$, and $u1\in\dom(t)$ implies
$u0\in\dom(t)$ (note that we allow the empty tree $\emptyset$ with
$\dom(\emptyset)=\emptyset$). A \emph{(bottom up) tree automaton} is a
tuple $\cA=(Q,\iota,\Delta,F)$ where $Q$ is a finite set of
\emph{states}, $\iota$ is the initial state, $\Delta\subseteq
Q\times\Sigma \times Q^2$ is the \emph{transition relation}, and
$F\subseteq Q$ is the set of \emph{final states}. A \emph{run} of the
tree automaton $\cA$ on the tree $t$ is a mapping $\rho:\dom(t)\to Q$
such that
\[
  (\rho(u),t(u),\rho'(u0),\rho'(u1))\in\Delta\text{ with }
  \rho'(v)=
  \begin{cases}
   \rho(v) & \text{ for } v\in\dom(t)\\
   \iota & \text{ otherwise} 
  \end{cases}
\]
holds for all $u\in\dom(t)$. The run $\rho$ is \emph{accepting} if
$\rho(\varepsilon)\in F$. The \emph{language of the tree
  automaton~$\cA$} is the set $L(\cA)$ of all trees $t$ that admit an
accepting run of~$\cA$ on~$t$. A set~$L$ of trees is \emph{regular} if
there exists a tree automaton $\cA$ with $L(\cA)=L$.

It is convenient to understand a \emph{word} as a tree $t$ with
$\dom(t)\subseteq 0^*$ (then $t(\varepsilon)$ is the first letter of
the word). Nevertheless, we will use standard notation for words like
$uv$ for the concatenation or $\varepsilon$ for the empty word. A
\emph{word automaton} is a tree automaton $\cA=(Q,\iota,\Delta,F)$
with
\[
  (q,a,p_0,p_1)\in\Delta\ \Longrightarrow\ 
    p_1=\iota\text{ and }q\neq\iota\,.
\]
This condition ensures that word automata accept words, only.

Let $t_1,\dots,t_n$ be trees and let $\#\notin\Sigma$.  Then
$\Sigma_\#=\Sigma\cup\{\#\}$ and the \emph{convolution}
$\otimes(t_1,t_2,\dots,t_n)$ or $t_1\otimes t_2\otimes\cdots\otimes
t_n$ is the $\Sigma_\#^n$-tree $t$ with $\dom(t)=\bigcup_{1\le i\le
  n}\dom(t)$ and
\[
  t(u)=(t_1'(u),t_2'(u),\dots,t_n'(u)) \text{ with }
  t_i'(u)=
  \begin{cases}
    t_i(u) & \text{ if }u\in\dom(t_i)\\
    \# & \text{ otherwise.}
  \end{cases}
\]
Note that the convolution of a tuple of words is a word, again. For an
$n$-ary relation $R$ on the set of all trees, we write $R^\otimes$ for
the set of convolutions $\otimes(t_1,\dots,t_n)$ with
$(t_1,\dots,t_n)\in R$. A relation $R$ on trees is \emph{automatic} if
$R^\otimes$ is a regular tree language.

A relational structure $\cS=(L;R_1,\dots,R_k)$ is \emph{tree
  automatic} if the tree languages $L$ and $R_i^\otimes$ for $1\le
i\le k$ are regular; it is \emph{word automatic} if, in addition, $L$
is a word language. A tuple of tree automata accepting $L$ and
$R_i^\otimes$ for $1\le i\le k$ is called a \emph{tree or word
  automatic presentation} of the structure~$\cS$.

\subsection{Linear orders}

For words $u$ and $v$, we write $u\le_\pref v$ if $u$ is a prefix
of~$v$. Let $\Sigma$ be some set linearly ordered by~$\le$. Then
$\le_\lex$ denotes the lexicographic order on the set of
words~$\Sigma^*$: $u\le_\lex v$ if $u\le_\pref v$ or there are
$x,y,z\in\Sigma^*$, $a,b\in\Sigma$ with $u=xay$, $v=xbz$, and
$a<b$. From the lexicographic order on~$\Sigma^*$, we derive a linear
order (denoted $\le_\lex^2$) on the set $\Sigma^*\otimes\Sigma^*$ of
convolutions of words by
\[
  u\otimes v \le_\lex^2 u'\otimes v'
  \mathrel{:\Leftrightarrow} u<_\lex u' \text{ or }u=u', v\le_\lex v'\,.
\]

By $\le_\llex$, we denote the \emph{length-lexicographic order}
defined by $u\le_\llex v$ if $|u|<|v|$ or $|u|=|v|$ and $u\le_\lex
v$. We next extend this linear order $\le_\llex$ to trees. Let $t$ be
a tree. Then $t\rest_{0^*}$ (more precisely,
$t\rest_{(0^*\cap\dom(t))}$) is a word that can be understood as the
``main branch'' of the tree~$t$. For $u\in\{0,1\}^*$, let $t\rest_u$
denote the subtree of $t$ rooted at~$u$ (i.e., $\dom(t\rest_u)=\{v\mid
uv\in\dom(t)\}$ and $t\rest_u(v)=t(uv)$ for $u\in\{0,1\}^*$ as well as
$t\rest_u=\emptyset$ for $u\notin\dom(t)$). Furthermore, $\tau(t)$ is
the tuple of ``side trees'' of~$t$, namely
\[
  \tau(t)=(t\rest_{0^i1})_{0^i\in\dom(t)}\,.
\]
We now define the extension $\le_\trees$ of $\le_\llex$ to trees
setting $s<_\trees t$ if and only if
\begin{itemize}
\item $s$ is the empty tree or
\item $s\rest_{0^*} <_\llex t\rest_{0^*}$ or
\item $s\rest_{0^*} = t\rest_{0^*}$ and there exists $i$ (with
  $0^i\in\dom(s)$) such that $s\rest_{0^j1}=t\rest_{0^j1}$ for all
  $0\le j<i$ and $s\rest_{0^i1} <_\trees t\rest_{0^i1}$.
\end{itemize}
In other words, we first compare the main branches of the trees~$s$
and $t$ length-lexicographically and, if they are equal, compare the
tuples $\tau(s)$ and $\tau(t)$ (length-)lexicographically (based on
the extension~$\le_\trees$ of the length-lexicographic order to
trees). Since the ``side trees''~$t\rest_{0^j1}$ of any tree~$t$ are
properly smaller than the tree itself, the relation $\le_\trees$ is
well-defined. Note that all the order relations $\le_\pref$,
$\le_\lex$, $\le_\lex^2$, $\le_\llex$, and $\le_\trees$ are automatic.
\medskip

A linear order $\cL$ is \emph{scattered} if there is no embedding of
the rational line $(\bQ;\le)$ into~$\cL$. Examples of scattered linear
orders are the linear order of the non-negative integers~$\omega$, of
the non-positive integers~$\omega^*$, or the linear order of
size~$n\in\bN$ that we denote~$\underline{n}$. If $\Sigma$ is an
alphabet with at least~2 letters, then
$(\Sigma^*;\le_\llex)\cong\omega$ is scattered, too. On the other
hand, if $a,b\in\Sigma$ are distinct letters, then
$(\{aa,bb\}^*ab;\le_\lex)\cong(\bQ;\le)$. Hence $(\Sigma^*;\le_\lex)$
is not scattered. From \cite[Prop.~4.10]{KhoRS05}, we know that the
set of word automatic presentations of scattered linear orders is
decidable.

A linear order $\cL=(L;\le)$ is \emph{rigid} if it does not admit any
non-trivial automorphism, i.e., if the identity mapping $f:L\to
L:x\mapsto x$ is the only automorphism of~$\cL$. The linear orders
$\omega$, $\omega^*$, and $\underline n$ for $n\in\bN$ are all
rigid. On the other hand, $(\bQ;\le)$ as well as $(\mathbb Z;\le)$ are
not rigid.

Note that automorphisms of tree automatic linear orders are binary
relations on $\Sigma^*$. Hence it makes sense to speak of an
\emph{automatic automorphism}. An automatic structure is
\emph{automatically rigid} if it does not have any non-trivial
automatic automorphisms.

 \medskip

Let $\cI=(I;\le)$ be a linear order and, for $i\in I$, let
$\cL_i=(L_i;\le_i)$ be a linear order. Then the
\emph{$\cI$-sum}\footnote{Shuffle sums mentioned in the introduction
  are special cases of this construction where $\cI=(\bQ;\le)$ is the
  rational line and, for every $q\in\bQ$, the set $\{r\in\bQ\mid
  \cL_q\cong\cL_r\}$ is dense.} of these linear orders is defined by
\[
  \sum_{i\in(I;\le)} \cL_i=\left(\biguplus_{i\in I} L_i;
       \bigcup_{i\in I}\mathord{\le_i}\cup
       \bigcup_{\substack{i,j\in I\\i<j}}(L_i\times L_j)\right)\,.
\]
For $\sum_{i\in\underline 2}\cL_i$, we simply write $\cL_1+\cL_2$.
If, for all $i\in I$, $\cL_i=\cL$, then we write $\cL\cdot\cI$ for
$\sum_{i\in(I;\le)}\cL_i$. Note that $\cL\cdot\cI$ is obtained by
replacing every element of~$\cI$ by a copy of~$\cL$. As an example,
consider the linear order $\delta=\omega\cdot\omega^*$. This linear
order will be used as ``delimiter'' in our constructions. It is
isomorphic to $(\bN\times\bN;\le_\delta)$ with
\[
  (i,j)\le_\delta(k,\ell)\mathrel{:\Leftrightarrow}
    j>\ell \text{ or }j=\ell\text{ and }i\le k\,.
\]
Hence it forms a descending chain of ascending chains. Therefore, it
has no minimal and no maximal element, is rigid and scattered. Note that
\[
   \delta\cong(10^+1^+0;\le_\lex)
\]
where we assume $0<1$. The isomorphism is given by $(i,j)\mapsto
10^{j+1}1^{i+1}0$.

\label{page-definition-D}
Also for later use, we next define a regular set $D=\{t_{i,j}\mid
i,j\ge0\}$ of trees such that $\delta\cong(D;\le_\trees)$. The
alphabet of these trees will be the singleton $\{\$\}$ so that a tree
is completely given by its domain. Then set inductively
\begin{align*}
  \dom(t_{0,j})& =\{\varepsilon, 0, 00\}\cup 
                  1\{0^k\mid 0\le k\le j\} \text { and } \\
  \dom(t_{i+1,j}) & =\{\varepsilon, 0, 00\}\cup  01\,\dom(t_{i,j})
\end{align*}

The trees $t_{0,4}$ and $t_{2,2}$ are depicted in
Figure~\ref{fig:trees-from-D} (left-arrows denote $0$-sons,
right-arrows denote $1$-sons). The tuple $\tau(t_{i,j})$ has the
following form
\[
  \tau(t_{0,j})=(t_j,\emptyset,\emptyset)
\text{ with } \dom(t_j)=\{0^k\mid 0\le k\le j\}\text{ and }
  \tau(t_{i+1,j}) = (\emptyset,t_{i,j},\emptyset)\,.
\]
Note that all trees $t_{i,j}$ coincide on their main branch, i.e.,
$t_{i,j}\rest_{0^*}= t_{k,\ell}\rest_{0^*}$. Hence $t_{i,j}\le_\trees
t_{k,\ell}$ if and only if $\tau(t_{i,j})$ is lexicographically
smaller than $\tau(t_{k,\ell})$. But this is the case if and only if
\begin{itemize}
\item $0=k<i$ or
\item $0=i=k$ and $j\le\ell$ or
\item $0<i,k$ and $t_{i-1,j}\le_\trees t_{k-1,\ell}$.
\end{itemize}
By induction, this is equivalent to $i>k$ or $i=k$, $j\le \ell$. Hence
$(D;\le_\trees)\cong\delta$. 

\begin{figure}
  \centering
  {
    \begin{picture}(30,50)(-5,5) 
      \gasset{Nw=2,Nh=2}
      \node(2)(15,50){} \node(3)(10,45){} \node(4)(5,40){}
      \drawedge(2,3){} \drawedge(3,4){}
      \node(5)(20,45){} \drawedge(2,5){} \node(6)(15,40){}
      \drawedge(5,6){} \node(7)(10,35){} \drawedge(6,7){}
      \node(8)(5,30){} \drawedge(7,8){} \node(9)(0,25){}
      \drawedge(8,9){}
    \end{picture}
    \begin{picture}(50,50)(-10,5)
      \gasset{Nw=2,Nh=2}
      \node(2)(15,50){} \node(3)(10,45){} \node(4)(5,40){}
      \drawedge(2,3){} \drawedge(3,4){}

      \node(5)(15,40){} \drawedge(3,5){}
      \node(6)(10,35){} \drawedge(5,6){}
      \node(7)(5,30){} \drawedge(6,7){}

      \node(5)(15,30){} \drawedge(6,5){}
      \node(6)(10,25){} \drawedge(5,6){}
      \node(7)(5,20){} \drawedge(6,7){}

      \node(5')(20,25){} \drawedge(5,5'){}
      \node(6)(15,20){} \drawedge(5',6){}
      \node(7)(10,15){} \drawedge(6,7){}
    \end{picture}
  }
  \caption{Two trees from $D$}
\label{fig:trees-from-D}
\end{figure}
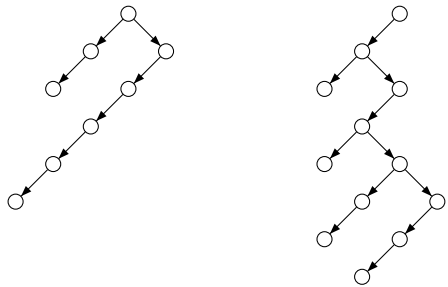

\section{Automorphisms of linear orders on words}

In this section, we consider linear orders on sets of words. The
universe will be regular or contextfree and the order will mainly be
the lexicographic order $\le_\lex$ and its relative~$\le_\lex^2$.

\subsection{Regular universe and $\le_\lex$}

Courcelle~\cite{Cou78} initiated the study of regular words, i.e.,
labeled linear orders derived from frontiers of regular trees. Thomas
proved that the isomorphism problem for these words is decidable
\cite{Tho86}, the complexity of this problem was determined by Lohrey
and Mathissen~\cite{LohM11}.

Based on techniques and results from \cite{BloE05}, we will show that,
given a regular language~$L$, it is decidable whether $(L;\le_\lex)$
is rigid. This proof requires the consideration of regular words: An
\emph{extended word} is a labeled linear order with a finite set of
labels. A \emph{regular word} over the alphabet~$A$ is an extended
word $(L;\le,\lambda)$ with $\lambda:L\to A$ such that
\begin{itemize}
\item $L$ and $\lambda^{-1}(a)$ for $a\in A$ are regular subsets of
  $\Sigma^*$ and
\item $\le$ is the lexicographic linear order $\le_\lex$.
\end{itemize}

Regular words can be described by \emph{terms over $A$} that we define
next. These terms use constants $a\in A$ (standing for the extended
word on~$\underline1$ whose only element is labeled~$a$) and the
following operations:
\begin{itemize}
\item concatenation of words (denoted $\mu+\nu$)
\item $\omega$-power (denoted $\mu\cdot\omega$)
\item $\omega^*$-power (denoted $\mu\cdot\omega^*$)
\item shuffle (denoted $[\nu_1,\nu_2,\dots,\nu_k]^\eta$) for arbitrary
  $k\ge1$. 
\end{itemize}
The semantics of the concatenation, $\omega$-power and
$\omega^*$-power generalize the corresponding operations for linear
orders in the obvious way. To define the extended word
$[\nu_1,\dots,\nu_k]^\eta$, let $\lambda:\bQ\to\{1,2,\dots,k\}$ be a
mapping such that $\lambda^{-1}(i)$ is dense for all $1\le i\le
k$. Then set $\nu(q)=\nu_{\lambda(q)}$ for $q\in\bQ$ and define
\[
  [\nu_1,\dots,\nu_k]^\eta=\sum_{q\in(\bQ;\le)}\nu(q)
\]
as we did for linear orders. For a term $t$, let $|t|$ denote the
extended word it describes. 

Let $\nu=(L;\le_\lex,\lambda)$ be a regular word given by finite
automata that accept $L$ and $\lambda^{-1}(a)$ for $a\in A$ (without
loss of generality, we can assume $\varepsilon\notin L$). Let
$\mathrm{Pref}(L)\subseteq\Sigma^*$ denote the set of proper prefixes
of words from~$L$. Then $T=(\mathrm{Pref}(L)\cup\{w\$a\mid wa\in
L\};\le_\pref)$ is a regular tree whose leaves are of the form $w\$a$
for $wa\in L$. Let $\$$ be the least letter of
$\Sigma\cup\{\$\}$. Then we can recover $\nu$ by reading the leaves of
the tree from left to right and label them by their last letter. From
this regular tree, we can read off a system of equations as follows:
Let $T_1, T_2, \dots, T_n$ denote the subtrees of $T$ (up to
isomorphism) with $T=T_1$. We have $n$ variables $x_1,x_2,\dots,x_n$
and the system of the following equations: if $T_i$ is not a
singleton, then we have the equation
\[
   x_i= x_{i_1} x_{i_2} \dots x_{i_{k_i}}
\]
where $T_{i_1}, T_{i_2} \dots T_{i_{k_i}}$ are the subtrees of $T_i$
rooted at the children of the root. If $T_i$ is a singleton, and if
its only node is labeled $a\in A$, then we have the equation
\[
   x_i = a\,.
\]
Then the regular word $\nu$ is the initial solution (in the sense of
\cite{Cou78}) of this system of equations.

From this system of equations, one can compute a term $t$ with
$|t|\cong\nu$ (Heilbrunner~\cite{Hei80}). Thus, to decide whether
$\nu$ has a nontrivial automorphism, we have to be able to decide,
given a term $t$, whether $|t|$ has a nontrivial automorphism.

Let $\nu=(L;\le,\lambda)$ be an extended word. On the set $L$, we
define an equivalence relation~$\sim$ by $x\sim y$ if (where we assume
$x\le y$)
\begin{itemize}
\item the interval $[x,y]$ is finite or
\item for any $x',y',z\in[x,y]$ with $x'<y'$, there exists
  $z'\in(x',y')$ with $\lambda(z)=\lambda(z')$.
\end{itemize}
Bloom and \'Esik \cite{BloE05} define a (decidable) class of
terms~$D(A)$ (called \emph{primitive terms in normal form}) with the
following properties
\begin{itemize}
\item If $\nu$ is a regular word with a single $\sim$-equivalence
  class, then there exists a term $t\in D(A)$ with $\nu\cong|t|$.
\item If $t\in D(A)$, then $|t|$ has a single $\sim$-equivalence class.
\item If $s,t\in D(A)$ with $|s|\cong|t|$, then $s=t$.
\end{itemize}

Let $\nu=(L;\le_\lex,\lambda)$ be a regular word. The equivalence
classes with respect to~$\sim$ are convex sets. Hence they can be ordered by
\[
   [x]_\sim <' [y]_\sim\mathrel{:\Leftrightarrow} 
       x<y\text{ and }x\not\sim y
\]
such that $(L/\mathord\sim;\le')$ is a linear order. For $X\in
L/\mathord\sim$, the restriction of $\nu$ to the equivalence class $X$
is a regular word with a single $\sim$-equivalence class. Hence there
exists a unique term $t_X\in D(A)$ with $|t|\cong\nu\rest X$. Define
$\lambda':L/\mathord\sim\to D(A)$ by $X\mapsto t_X$. Then
\[
  c(\nu)=(L/\mathord\sim;\le',\lambda')
\]
is an extended word with possibly infinite alphabet. 

To decide whether $|t|$ is rigid, we proceed as follows: Using the
algorithm by Bloom and \'Esik \cite{BloE05}, we construct a term
$c(t)$ with $|c(t)|\cong c(|t|)$, in particular, $c(|t|)$ has a finite
alphabet. From this term $c(t)$, we can extract the set $D$ of terms
from $D(A)$ that appear in~$c(t)$. Then we observe that $|t|$ has a
nontrivial automorphism if and only if
\begin{itemize}
\item $c(|t|)=|c(t)|$ has a nontrivial automorphism or
\item there exists a $\sim$-equivalence class $X$ such that $|t|\rest
  X$ has a non-trivial automorphism.
\end{itemize}
Note that $s\in D$ if and only if there exists a $\sim$-equivalence
class $X$ with $|t|\rest X\cong|s|$. Hence the second item holds if
and only if there exists $s\in D$ such that $|s|$ has a nontrivial
automorphism -- but this is the case if and only if $s$ is of the form
$u\cdot\omega^*+v\cdot\omega$ or $[u_1,\dots,u_k]^\eta$.  To decide
whether $|c(t)|$ has a nontrivial automorphism, we call this process
recursively. From \cite{KhoRS05}, we observe that $c^n(|t|)$ is a
singleton for some~$n\in\bN$, hence this recursive procedure stops
eventually with $|t|$ a singleton.

Formulated for linear orders, we therefore showed
\begin{theorem}\label{T-regular-rigid}
  The set of regular languages $L$ such that $(L;\le_\lex)$ is rigid,
  is decidable.
\end{theorem}

\subsection{Regular universe and $\le_\lex^2$}\label{S-word-rigid}

The situation changes completely when we move from the lexicographic
order $\le_\lex$ to the linear order $\le_\lex^2$ since, as we will
see, rigidity of $(L;\le_\lex^2)$ is undecidable for regular
languages~$L$. 

Let $p,q\in\bN[\bar x]$ be two polynomials with coefficients in $\bN$
and variables among $\bar x=(x_1,\dots,x_k)$. Then define the linear order 
\[
   \cL_{p,q}=\sum_{\bar x\in(\bN^k;\le_\lex)}
     \left(
       (\underline{p(\bar x)}+\delta)\cdot\omega^*+(\underline{q(\bar x)}+\delta)\cdot\omega)
     \right)\,.
\]
This linear order $\cL_{p,q}$ forms an $\omega$-sequence of ``blocks''
of the form
\[
   B(m,n)=(\underline{m}+\delta)\cdot\omega^*+(\underline{n}+\delta)\cdot\omega
\]
with $m,n\in\bN$. Therefore, every automorphism of $\cL_{p,q}$ has to
map every block onto itself. In other words, $\cL_{p,q}$ is rigid if
and only if all these blocks are rigid. But $B(m,n)$ is rigid if and
only if $m\neq n$. Hence we showed
\begin{equation}
  \label{eq:auto1}
  \cL_{p,q}\text{ is rigid}\iff\forall\bar x\in\bN^k:p(\bar x)\neq q(\bar x)\,.
\end{equation}
Finally note that $\cL_{p,q}$ is scattered since $\delta$, $\omega$, and
$\omega^*$ are all scattered.

We now prove that $\cL_{p,q}$ is automatic or, more specifically, we
will construct a regular set $L\subseteq\{0,1\}^+\otimes\{0,1\}^+$
such that $\cL_{p,q}\cong(L;\le_\lex^2)$ (see
Lemma~\ref{L-word-automatic} below).

Let $\cA=(Q,\iota,\Delta,F)$ be a word automaton over the
alphabet~$\Sigma$ and let $w\in\Sigma^+$ be a word. Then $\Run(\cA,w)$
is the set of all words over $\Delta$ of the form
\[
  (q_0,a_1,q_1,\iota)(q_1,a_2,q_2,\iota)\dots(q_{k-1},a_k,\iota,\iota)
\]
with $w=a_1a_2\dots a_k$ and $q_0\in F$. These words encode the
accepting runs of the word automaton~$\cA$ (recall that word automata
are special bottom up tree automata which explains the unusual
position of the initial and final states in the run). Furthermore, let
$\Run(\cA)=\bigcup_{w\in\Sigma^+}\Run(\cA,w)$.

\begin{lemma}\label{L-word-automatic-1}
  From polynomials $p,q\in\bN[x_1,\dots,x_k]$, one can construct an
  alphabet $\Sigma$ and a regular language
  $K\subseteq\Sigma^+\otimes\Sigma^+$ such that
  $(K;\le_\lex^2)\cong\cL_{p,q}$.

  If $\cL_{p,q}$ has a non-trivial automorphism, $(K;\le_\lex^2)$ has
  a non-trivial automatic automorphism.
\end{lemma}

\begin{proof}
  Let $p$ and $q$ be polynomials from $\bN[x_1,\dots,x_k]$. For $\bar
  x=(x_1,\dots,x_k)\in\bN^k$, set 
  \[
  a^{\bar x}=a^{x_1}\cent a^{x_2}\cent \cdots\cent a^{x_k}\cent 
         \in(a^*\cent )^k\,.
  \]
  Then, as in the proof of \cite[Lemma~7]{KusLL11}, one can construct
  nondeterministic finite automata~$\cA_p=(Q_p,\iota_p,\Delta_p,F_p)$
  and $\cA_q=(Q_q,\iota_q,\Delta_q,F_q)$ with
  $L(\cA_p),L(\cA_q)\subseteq(a^*\cent )^k$, such that, for~$\bar
  x\in\bN^k$, the NFA~$\cA_p$ has precisely $p(\bar x)$ many accepting
  runs on the word~$a^{\bar x}$, i.e., $|\Run(\cA_p,a^{\bar
    x})|=p(\bar x)$, and similarly $|\Run(\cA_q,a^{\bar x})|=q(\bar
  x)$. We will assume $\Delta_p\cap\Delta_q=\emptyset$.

  Define the language $K$ by
  \begin{eqnarray*}
  K&=&\bigcup_{\bar x\in\bN^k}a^{\bar x}0^+1\otimes(\Run(\cA_p,a^{\bar
    x})\cup 32^+3^+2)\\&&\cup\bigcup_{\bar x\in\bN^k}a^{\bar x}1^+0\otimes(\Run(\cA_q,a^{\bar
    x})\cup 32^+3^+2)\,.
  \end{eqnarray*}
  Hence any word from $K$ is the convolution of two words over the
  alphabet 
  \[
    \Sigma=\{a,\cent ,0,1,2,3\}\cup\Delta_p\cup\Delta_q\,.
  \]
  
  We have to show that the language~$K$ is effectively regular. Here,
  the crucial point is the regularity of
  \[
    \bigcup_{\bar x\in\bN^k} a^{\bar x}0^+1\otimes\Run(\cA_p,a^{\bar x})
    =
    \left[\bigcup_{\bar x\in\bN^k} a^{\bar x}\otimes\Run(\cA_p,a^{\bar x})
    \right]
    \cdot
    \left(0^+1\otimes\{\varepsilon\}\right)
  \]
  (this equality holds since $|w|=|W|$ for any $w\in(a^*\cent )^k$ and
  $W\in\Run(\cA_p,w)$).  But a word belongs to the language in square
  brackets if and only if it is the convolution of a word~$w$ from the
  regular language $(a^*\cent )^k$ and a run of the automaton $\cA_p$
  on this word~$w$, a property that a finite automaton can check
  easily.

  On the alphabet~$\Sigma$, we now fix a linear order~$\le$ such that
  \[
  \Delta_p\cup\Delta_q < 0 < 1 < 2 < 3 < \cent  < a\,.
  \]
  The associated order $\le_\lex^2$ on the language~$K$ can now be
  characterized as follows:
  \[
    a^{\bar x}b^m (1-b)\otimes r \le_\lex^2 a^{\bar y}c^n(1-c)\otimes s
  \]
  (with $b,c\in\{0,1\}$) if and only if
  \begin{enumerate}[(i)]
  \item $b=0$, $c=1$ and $\bar x \le_\lex\bar y$, or
  \item $b=1$, $c=0$ and $\bar x<_\lex \bar y$, or
  \item $b=c$ and
    \begin{enumerate}[({iii}.1)]
    \item $\bar x<_\lex \bar y$, or
    \item $\bar x=\bar y$, $b=0$, and $m>n$, or
    \item $\bar x=\bar y$, $b=1$, and $m<n$, or
    \item $\bar x=\bar y$, $m=n$, and
      \begin{enumerate}[({iii.4}.1)]
      \item $r\in\Run(\cA_p)\cup\Run(\cA_q)$ and $s\in 32^+3^+2$, or
      \item $r,s\in\Run(\cA_p)\cup\Run(\cA_q)$ and $r \le_\lex s$, or
      \item $r,s\in 32^+3^+2$ and $r \le_\lex s$.
      \end{enumerate}
    \end{enumerate}
  \end{enumerate}

  We show $(K;\le_\lex^2)\cong\cL_{p,q}$. For $\bar x\in \bN^k$ and
  $m\ge1$, let $\cI_{\bar x,0^m1}$ denote the restriction of
  $(K;\le_\lex^2)$ to the set $a^{\bar x}0^m1\otimes
  (\Run(\cA_p,a^{\bar x})\cup 32^+3^+2)\subseteq K$. By (iii.4.1),
  $\cI_{\bar x,0^m1}$ is isomorphic to the sum of the restrictions of
  $(K;\le_\lex^2)$ to the sets $a^{\bar x}0^m1\otimes
  \Run(\cA_p,a^{\bar x})$ and $a^{\bar x}0^m1\otimes 32^+3^+2$,
  resp. By (iii.4.2) and the choice of the automaton $\cA_p$, the
  first restriction is isomorphic to~$\underline{p(\bar x)}$. Recall
  that $(32^+3^+2;\le_\lex)\cong\delta$. Hence, the second restriction
  is isomorphic to~$\delta$ by (iii.4.3). In summary,
  \[
  \cI_{\bar x,0^m1}\cong\underline{p(\bar x)}+\delta\,.
  \]
  Next, let $\cI_{\bar x,0^+1}$ denote the restriction of
  $(K;\le_\lex^2)$ to the set $a^{\bar x}0^+1 \otimes
  (\Run(\cA_p,a^{\bar x})\cup 32^+3^+2)$. Note that, by (iii.2), we
  have
  \[
  \cI_{\bar x,0^{m+1}1}<_\lex^2 \cI_{\bar x,0^m1}
  \]
  for all $m\ge1$. Hence
  \[
  \cI_{\bar x,0^+1}=\sum_{m\le -1}\cI_{\bar x,0^{-m}1} \cong
  (\underline{p(\bar x)}+\delta)\cdot\omega^*\,.
  \]

  With $\cI_{\bar x,1^+0}$ the restriction of $(K;\le_\lex^2)$ to the
  set $a^{\bar x}1^+0\otimes(\Run(\cA_q,a^{\bar x})\cup 32^+3^+2)$, we
  obtain similarly
  \[
  \cI_{\bar x,1^+0}=\sum_{m\ge 1}\cI_{\bar x,1^{m}0} \cong
  (\underline{q(\bar x)}+\delta)\cdot\omega
  \]
  (the reason for the factor $\omega$ instead of $\omega^*$ above is
  the difference between (iii.2) and (iii.3)).

  Finally, for $\bar x\in\bN^k$, let $\cI_{\bar x}$ denote the
  restriction of $(K;\le_\lex^2)$ to the set
  \[
     \bigg[a^{\bar x}0^+1\otimes(\Run(\cA_p,a^{\bar x})\cup 32^+3^+2)\bigg]
     \cup
     \bigg[a^{\bar x}1^+0\otimes(\Run(\cA_q,a^{\bar x})\cup 32^+3^+2)\bigg]\,.
  \]
  Then (i) with $\bar x=\bar y$ and the above imply
  \[
    \cI_{\bar x}=\cI_{\bar x,0^+1}+\cI_{\bar x,1^+0}
    \cong
    (\underline{p(\bar x)}+\delta)\cdot\omega^*
    +(\underline{q(\bar x)}+\delta)\cdot\omega\,.
  \]
  Together with (i), (ii), and (iii.1), this ensures
  \[
    (K;\le_\lex^2)=\sum_{\bar x\in (\bN^k;\le_\lex)} 
    \cI_{\bar x}\cong\cL_{p,q}\,.
  \]

  Now suppose that $\cL_{p,q}$ has a non-trivial automorphism. Then,
  as we saw above, there is $\bar y\in\bN^k$ such that $p(\bar
  y)=q(\bar y)$. From the construction of the automata $\cA_p$ and
  $\cA_q$, we infer $|\Run(\cA_p,a^{\bar y})|=|\Run(\cA_q,a^{\bar
    y})|$. Let
  \[
    \Run(\cA_p,a^{\bar y})=\{\rho_1,\dots,\rho_n)\text{ and }
    \Run(\cA_q,a^{\bar y})=\{\sigma_1,\dots,\sigma_n)
  \]
  with
  \[
    \rho_1 <_\lex \rho_2 <_\lex \cdots <_\lex \rho_n \text{ and }
    \sigma_1 <_\lex \sigma_2 <_\lex \cdots <_\lex \sigma_n\,.
  \]
  Now define a mapping $f:K\to K$ by
  \[
    f(a^{\bar x}b^m(1-b)\otimes r)=
    \begin{cases}
      a^{\bar x}b^m(1-b)\otimes r & \text{ if }\bar x\neq\bar y\\
      a^{\bar y}b^{m-1}(1-b)\otimes r & \text{ if }\bar x=\bar y, b=0, m>1\\
      a^{\bar y}10\otimes r & \text{ if }\bar x=\bar y, b=0, m=1, r\in 32^+3^+2\\
      a^{\bar y}10\otimes \sigma_i & \text{ if }\bar x=\bar y, b=0, m=1, r=\rho_i\\
      a^{\bar y}b^{m+1}(1-b)\otimes r & \text{ if }\bar x=\bar y, b=1
    \end{cases}
  \]
  This mapping fixes all elements of $K$ not belonging to $\cI_{\bar
    y}$. On this linear order $\cI_{\bar y}$, it acts as an
  automorphism. Hence $f$ is a non-trivial automorphism
  of~$(K;\le_\llex^2)$. Note that the universe of $\cI_{\bar y}$ is
  regular. It follows that $f^\otimes$ is regular.
\end{proof}

\begin{lemma}\label{L-word-automatic}
  From polynomials $p,q\in\bN[x_1,\dots,x_k]$, one can construct a
  regular language $L\subseteq\{0,1\}^+\otimes\{0,1\}^+$ such that
  $(L;\le_\lex^2)\cong\cL_{p,q}$.

  If $\cL_{p,q}$ has a non-trivial automorphism, then there exists a
  non-trivial automorphism~$h$ of $(L;\le_\lex^2)$ such that
  $h^\otimes$ is regular.
\end{lemma}

\begin{proof}
  Let $p,q\in\bN[x_1,\dots,x_k]$ be polynomials, let $K$ be the
  language from Lemma~\ref{L-word-automatic-1}, and let $(\Sigma;\le)$
  be the sequence
  \[
    \sigma_1<\sigma_2<\dots<\sigma_\ell\,.
  \]
  Furthermore, let $g$ denote the monoid homomorphism from~$\Sigma^*$
  to $\{0,1\}^*$ defined by $g(\sigma_i)=1^i 0^{\ell-i}$ for $1\le
  i\le \ell$. Now set $L=\{g(u)\otimes g(v)\mid u\otimes v\in
  K\}$. Then $g$ is an isomorphism from $(K;\le_\lex^2)$ onto
  $(L;\le_\lex^2)$. Since all the words $g(\sigma_i)$ have the same
  length, the language $L$ is also regular.

  If $\cL_{p,q}$ has a non-trivial automorphism, then, by
  Lemma~\ref{L-word-automatic-1}, there is a non-trivial
  automorphism~$f$ of $(K;\le_\llex^2)$ such that $f^\otimes$ is
  regular. Hence $g=h^{-1}\circ f\circ h$ is a non-trivial
  automorphism of $(L;\le_\llex^2)$. Note that $h^\otimes$ is
  regular. It follows that also $ g^\otimes$ is regular.
\end{proof}

\begin{theorem}\label{T-rigid-automatic}
  \begin{enumerate}[(i)]
  \item The set of regular languages
    $L\subseteq\{0,1\}^+\otimes\{0,1\}^+$ such that $(L;\le_\lex^2)$
    is rigid (is rigid and scattered, resp.), is $\Pi^0_1$-hard.
  \item The set of regular languages
    $L\subseteq\{0,1\}^+\otimes\{0,1\}^+$ such that $(L;\le_\lex^2)$
    is automatically rigid (automatically rigid and scattered, resp.)
    is $\Pi^0_1$-hard.
  \end{enumerate}
\end{theorem}

\begin{proof}
  \begin{enumerate}[(i)]
  \item The set of pairs of polynomials $p,q\in\bN[\bar x]$ with
    $p(\bar y)\neq q(\bar y)$ for all $\bar y\in\bN^k$ is
    $\Pi^0_1$-complete~\cite{Mat93}. We reduce this to the first set
    in question: Let $p,q\in\bN[\bar x]$ and let $L$ be the regular
    language from Lemma~\ref{L-word-automatic}. Then
    $(L;\le_\lex^2)\cong\cL_{p,q}$ is rigid if and only if $p(\bar
    y)\neq q(\bar y)$ for all $\bar y\in\bN^k$ by \eqref{eq:auto1}.

    Note that this is even a reduction to the second set in question
    since the linear order $\cL_{p,q}$ is scattered.
  \item Be Lemma~\ref{L-word-automatic}, $\cL_{p,q}$ is rigid if and
    only if $(L;\le_\llex^2)$ is automatically rigid. Hence the above
    reduction also proves the two claims from (ii).
  \end{enumerate}
\end{proof}

\begin{corollary}\label{C-rigid-automatic}
  \begin{enumerate}[(i)]
  \item The set of word automatic presentations of rigid (rigid and
    scattered, resp.) linear orders is $\Pi^0_1$-hard.
  \item The set of word automatic presentations of automatically rigid
    (automatically rigid and scattered, resp.) linear orders is
    $\Pi^0_1$-complete.
  \end{enumerate}
\end{corollary}

\begin{proof}
  The two claims from (i) are obvious consequences of
  Theorem~\ref{T-rigid-automatic}(i). Analogously, the two hardness
  claims from (ii) follow immediately from
  Theorem~\ref{T-rigid-automatic}(ii). 

  Now let $(L;\le)$ be an automatic linear order given by a word
  automatic presentation. Let $R\subseteq
  \Sigma^+\times\Sigma^+$. Then it can be expressed in first-order
  logic that $R$ is a non-trivial automorphism of $(L;\le)$. Hence,
  given a finite automaton $\cA$ for a regular language
  $R^\otimes\subseteq\Sigma^+\otimes\Sigma^+$, one can decide whether
  $R$ is a non-trivial automorphism of $(L;\le)$
  \cite{KhoN95}. Consequently, automatic rigidity of $(L;\le)$ is a
  $\Pi^0_1$-property.
\end{proof}

\subsection{Contextfree universe and \protect{$\le_\lex$}}

\'Esik initiated the investigation of linear orders of the form
$(L;\le_\lex)$ where $L$ is contextfree. Density of such a linear
order is undecidable \cite{Esi11}, the isomorphism problem is
$\Sigma^1_1$-complete~\cite{KusLL11} and their rank is bounded by
$\omega^\omega$ \cite{CarE12}.

We will show that rigidity of $(L;\le_\lex)$ is undecidable for
context-free languages~$L$. The proof uses the linear order
$\cL_{p,q}$ and constructs a deterministic context-free language $L'$
such that $(L';\le_\lex)\cong\cL_{p,q}$. This construction is a
variant of the construction in the proof of
Lemma~\ref{L-word-automatic-1}.

\begin{lemma}\label{L-cf}
  From polynomials $p,q\in\bN[x_1,\dots,x_k]$, one can construct a
  deterministic contextfree language $L'\subseteq\{0,1\}^+$ such that
  $(L';\le_\lex)\cong\cL_{p,q}$.
\end{lemma}

\begin{proof}
  Let $p,q\in\bN[x_1,\dots,x_k]$ be polynomials and let $K$ be the
  language from Lemma~\ref{L-word-automatic-1}. Then set
  \[
    K'=\{u\$v^{rev}\mid u\otimes v\in L\}
  \]
  where $v^{rev}$ is the reversal of the word~$v$. Then, from a
  deterministic finite automaton~$\cA$ accepting~$K^{rev}$, one can
  construct a deterministic pushdown automaton accepting~$K'$ (reading
  $u\$v$, it stores $u$ in the stack and, after reading $\$$,
  simulates $\cA$ while emptying the stack). Note that the alphabet of
  $K'$ is
  \[
     \Sigma'=\{\$\}\cup\Sigma=\{\$,a,\cent ,0,1,2,3\}\cup\Delta_p\cup\Delta_q\,.
  \]
  We order the alphabet~$\Sigma'$ by~$\le'$ such that
  \[
    \Delta_p\cup\Delta_q <' 0 <' 1 <' 3 <' 2 <' \cent  <' a <' \$\,.
  \]
  Compared to the proof of Lemma~\ref{L-word-automatic-1}, the order
  of $2$ and $3$ is inverted and $\$$ is made the new maximal element
  (we could have placed $\$$ anywhere). With $\le$ the order on
  $\Sigma$ from the proof of Lemma~\ref{L-word-automatic-1}, one
  effect of this definition is
  $(32^+3^+2^{rev};\le'_\lex)\cong(32^+3^+2;\le_\lex)\cong\delta$
  which will be used below.

  To show $(K';\le'_\lex)\cong\cL_{p,q}$, is suffices to prove
  $(K';\le'_\lex)\cong(K;\le_\lex^2)$. For this, recall that
  $(K;\le_\lex^2)$ is a sequence of the following blocks (for $\bar
  x\in\bN^k$ and $m\ge1$):
  \begin{itemize}
  \item $(a^{\bar x}0^m1\otimes\Run(\cA_p,a^{\bar x});\le_\lex^2)$:
    This linear order is finite of size $|\Run(\cA_p,a^{\bar
      x})|$. The same holds of the linear order
    \[
      (a^{\bar x}0^m1\$\{r^{rev}\mid r\in\Run(\cA_p,a^{\bar x})\};\le'_\lex)\,.
    \]
  \item $(a^{\bar x}1^m0\otimes\Run(\cA_q,a^{\bar x});\le_\lex^2)$: As
    above, this is isomorphic to
    \[
      (a^{\bar x}1^m0\$\{r^{rev}\mid r\in\Run(\cA_q,a^{\bar x})\};\le'_\lex)\,.
    \]
  \item $(a^{\bar x}b^m(1-b)\otimes 32^+3^+2;\le_\lex^2)$ (for
    $b\in\{0,1\}$) which is isomorphic to $\delta$. But $\delta$ is
    also isomorphic to
    \[
      (a^{\bar x}b^m(1-b)\$23^+2^+3;\le'_\lex)
    \]
    as we saw above.
  \end{itemize}
  It therefore follows that $(K;\le_\lex^2)$ and $(K',\le'_\lex)$ are
  isomorphic. The construction of $L'\subseteq\{0,1\}^+$ then follows
  the proof of Lemma~\ref{L-word-automatic}.
\end{proof}

Now we obtain, in the same way that we proved
Theorem~\ref{T-rigid-automatic}, the following result.

\begin{theorem}\label{T-rigid-cf}
  The set of contextfree languages $L\subseteq\{0,1\}^+$ such that
  $(L;\le_\lex)$ is rigid (is rigid and scattered, resp.), is
  $\Pi^0_1$-hard.
\end{theorem}

\section{Isomorphisms and automorphisms of linear orders on trees}

In this section, we will show that the isomorphism of scattered and
tree automatic linear orders is undecidable. Furthermore, we will
prove that the existence of a non-trivial automorphism in this case is
$\Sigma^0_2$-hard. Both these results use (an improved version of) a
theorem by Krob~\cite{Kro94} that we discuss first.

\subsection{Weighted automata and Minsky machines}

A \emph{weighted automaton} is a tuple $\cA=(Q,\Sigma,\iota,\mu,F)$
where $Q$ is the finite set of states, $\Sigma$ the alphabet,
$\iota\in Q$ is the initial state, $F\subseteq Q$ is the set of
accepting states, and $\mu:Q\times\Sigma\times Q\to\{-\infty,0,1\}$ is
the weight function.

A \emph{run} of $\cA$ is a sequence
$\rho=)(q_0,a_1,q_1)\dots(q_{k-1},a_k,q_k)\in\Delta^+$ with
$q_{0}=\iota$, $\mu(q_{i-1},a_i,q_i)\in\{0,1\}$, and $q_k\in F$.  Its
\emph{label} is the word $a_1\dots a_k\in \Sigma^+$. By $\Run(\cA,w)$
we denote the set of runs labeled $w$ and $\Run(\cA)$ denotes the set
of all runs of~$\cA$. The \emph{weight}~$\wt(\rho)$ of the run~$\rho$
is the number of indices~$i$ with $\mu(q_{i-1},a_i,q_i)=1$. The
\emph{behaviour~$||\cA||$} of $\cA$ is the function from $\Sigma^+$ to
$\bN\cup\{-\infty\}$ that maps the word~$w$ to the maximal weight of a
run with label~$w$.

\begin{theorem}[cf.\ proof of \protect{\cite[Theorem~8.6]{DroK}}]\label{T-Krob}
  From a Minsky machine $\cM$, one can construct a weighted
  automaton~$\cA$ and a regular
  language~$\CTreg\subseteq(\Sigma\cdot\Box)^+$ such that, for any
  $m\in\bN$, the following are equivalent:
  \begin{enumerate}
  \item $m$ is not accepted by $\cM$.
  \item $||\cA||(u)>\frac12|u|$ for all $u\in\CTreg$ with $m=\max\{n\mid
    \$\Box(a\Box)^n\le_\pref u\}$.
  \end{enumerate}
  Furthermore, $||\cA||(u)\in\bN$ for all $u\in\CTreg$.
\end{theorem}

\begin{proof}
  This is a slight adaptation of the proof of
  \cite[Theorem~8.6]{DroK}. If $r$ is the function defined there, we
  construct the weighted automaton~$\cA$ such that
  \[
    ||\cA||(u)=
    \begin{cases}
      r(a_1a_2\dots a_k) & \text{ if }u=a_1\Box\,a_2\Box\,\dots\,a_k\Box\\
      0 & \text{ otherwise.}
    \end{cases}
  \]
  The reason for this modification is that here, transition weights
  are from $\{-\infty,0,1\}$ while, in \cite{DroK}, we also used the
  transition weight~$2$.  
\end{proof}

From the weighted automaton $\cA$, one can then construct (cf.\
\cite{DroKV09,DroK}) weighted automata~$\cA_\cM$ on the alphabet
$\Sigma$ and $\cB_\cM$ on the alphabet $\Sigma_{\#}^2$ such that
\begin{eqnarray}
  ||\cA_\cM||(u)&=&\max(\lfloor \textstyle\frac{|u|}{2}\rfloor+1,||\cA||(u))\text{ and }\label{eq:def-AM}\\
  ||\cB_\cM||(x)&=&
  \begin{cases}
    ||\cA||(u) & \text{if }
    \begin{array}[t]{l}
       x=u\otimes\$\Box(a\Box)^m, u\in\CTreg,\\\text{and }
       m=\max\{n\mid \$\Box(a\Box)^n\le_\pref u\}\\
     \end{array}\\
     ||\cA_\cM||(u) & \text{if }
     \begin{array}[t]{l}
       x=u\otimes\$\Box(a\Box)^m\text{ and}\\
       (u\notin\CTreg\text{ or }
           m\neq\max\{n\mid \$\Box(a\Box)^n\le_\pref u\})
   \end{array}\\
     0 & \text{otherwise}
  \end{cases}\label{eq:def-BM}
\end{eqnarray}
for all $u\in\Sigma^+$ and $x\in(\Sigma_{\#}^2)^+$.

For $m\in\bN$, we define the function $r_{\cM,m}:\Sigma^+\to\bN$ by
$r_{\cM,m}(u)=||\cB_\cM||(u\otimes \$\Box(a\Box)^m\}$.  This is
well-defined since, for any $u\in\Sigma^+$ and $m\in\bN$, we have
$||\cA||(u)\in\bN$ and therefore also
$||\cB_\cM||(u\otimes\$\Box(a\Box)^m)\in\bN$. In other words, we have
\begin{equation}
  \label{eq:r-Mm}
  r_{\cM,m}(u)=
  \begin{cases}
    ||\cA||(u) & \text{ if } u\in\CTreg\text{ and }
                    m=\max\{n\mid \$\Box(a\Box)^n\le_\pref u\}\\
    ||\cA_\cM||(u) & \text{ otherwise}
  \end{cases}
\end{equation}

\begin{proposition}\label{P-weighted-automata}
  For all $m\in\bN$, the following are equivalent:
  \begin{enumerate}
  \item $m$ is not accepted by the Minsky machine~$\cM$.
  \item $||\cA_\cM||(u)=r_{\cM,m}(u)$ holds
    for all $u\in\Sigma^*$.
  \end{enumerate}
\end{proposition}

\begin{proof}
  Suppose that $m$ is not accepted by the Minsky machine~$\cM$. Let
  $u\in\Sigma^*\setminus\CTreg$ or $m\neq\max\{n\mid
  \$\Box(a\Box)^n\le_\pref u\}$. Then \eqref{eq:r-Mm} immediately
  implies $||\cA_\cM||(u)=r_{\cM,m}(u)$. Next let $u\in\CTreg$ and
  $m=\max\{n\mid \$\Box(a\Box)^n\le_\pref u\}$.  Then, by
  Theorem~\ref{T-Krob}, we have $||\cA||(u)>\frac12|u|$ implying
  $||\cA||(u)\ge\lfloor\frac{1}{2}|u|\rfloor+1$ and therefore
  $||\cA_\cM||(u)=||\cA||(u)=r_{\cM,m}(u)$ by~\eqref{eq:def-AM}
  and~\eqref{eq:r-Mm}.

  For the other direction, assume that $m$ is accepted by the Minsky
  machine~$\cM$. Then, by Theorem~\ref{T-Krob}, there exists a word
  $u\in\CTreg$ with $m=\max\{n\mid \$\Box(a\Box)^n\le_\pref u\}$ such
  that $||\cA||(u)\le\frac12|u|$. Note that the word~$u$ has even
  length since $u\in\CTreg$. Hence
  $||\cA||(u)<\lfloor\frac12|u|\rfloor+1$. Consequently
  $||\cA_\cM||(u)=\lfloor\frac12|u|\rfloor+1>||\cA||(u)=r_{\cM,m}(u)$
  by~\eqref{eq:def-AM} and~\eqref{eq:r-Mm}.
\end{proof}

\subsection{Isomorphism}

For a function $r:\Sigma^+\to\bN$, we set
\[
    \cL_r=\sum_{w\in(\Sigma^+;\le_\llex)}(\omega^{r(w)+1}+\delta)\,.
\]
Since $(\Sigma^+;\le_\llex)\cong\omega$, this linear order is an
$\omega$-sequence of ordinals, separated by our delimiter~$\delta$. Hence
it is scattered. Furthermore, we obtain
\begin{equation}
  \label{eq:iso1}
  \cL_r\cong \cL_{r'}\iff r=r'
\end{equation}
for all functions $r,r':\Sigma^+\to\bN$. 

\begin{lemma}\label{L-weighted-automata}
  From a weighted automaton $\cA$, one can compute a regular set of
  trees $L_\cA$ such that $(L_\cA;\le_\trees)\cong\cL_{||\cA||}$.
\end{lemma}

Before we prove this lemma, we show how we can use it to prove that
the isomorphism problem of scattered tree automatic linear orders is
undecidable (the proof of Lemma~\ref{L-weighted-automata} can be found
following the proof of Corollary~\ref{C-Iso-tree-automatic}).

\begin{lemma}\label{L-tree-automatic-1}
  From a Minsky machine $\cM$ and $m\in\bN$, one can compute a regular
  set of trees $L$ such that $(L;\le_\trees)\cong\cL_{r_{\cM,m}}$.
\end{lemma}

\begin{proof}
  Let $\cM$ be a Minsky machine and let $m\in\bN$.  Let $\cB_\cM$ be
  the weighted automaton constructed following
  Theorem~\ref{T-Krob}. Then, from $m\in\bN$, we can compute a
  weighted automaton $\cB_{\cM,m}$ with alphabet $\Sigma$ such that
  \[
     ||\cB_{\cM,m}||(u)=||\cB_\cM||(u\otimes\$\Box(a\Box)^m)
        \text{ for all }u\in\Sigma^+\,.
  \]
  But then $||\cB_{\cM,m}||=r_{\cM,m}$.  By
  Lemma~\ref{L-weighted-automata}, we can compute, from $m\in\bN$, a
  regular language of trees $L$ such that
  $(L;\le_\trees)\cong\cL_{||\cB_{\cM,m}||}=\cL_{r_{\cM,m}}$.
\end{proof}

\begin{theorem}\label{T-Iso-tree-automatic}
  There is a scattered linear order $\cL$ such that the set of regular
  tree languages~$L$ with $(L;\le_\trees)\cong\cL$ is $\Pi^0_1$-hard.
\end{theorem}

\begin{proof}
  Let $P\subseteq\bN$ be some $\Pi^0_1$-complete set. Then there
  exists a Minsky machine $\cM$ that accepts the set $\bN\setminus
  P$. Let $\cA_\cM$ and $\cB_\cM$ be the weighted automata constructed
  following Theorem~\ref{T-Krob}. Then we get
  \begin{align}
    m\in P & \iff m\text{ is not accepted by }\cM\nonumber\\
           & \stackrel{\text{Prop.}~\ref{P-weighted-automata}}\iff 
               ||\cA_\cM||(u)=r_{\cM,m}(u)\text{ for all }u\in\Sigma^+
                  \label{eq:iso}\\
           & \iff \cL_{||\cA_\cM||}\cong\cL_{r_{\cM,m}} \nonumber
  \end{align}
  where the last equivalence follows from~\eqref{eq:iso1}. Hence, by
  Lemma~\ref{L-tree-automatic-1}, we can reduce the $\Pi^0_1$-complete
  set~$P$ to the set of regular tree languages~$L$ with
  $(L;\le_\trees)\cong\cL_{||\cA_\cM||}$. The theorem therefore holds
  with $\cL=\cL_{||\cA_\cM||}$.
\end{proof}

Since the linear order $\le_\trees$ is tree automatic, we immediately
obtain

\begin{corollary}\label{C-Iso-tree-automatic}
  There is a scattered linear order $\cL$ whose set of tree automatic
  presentations is $\Pi^0_1$-hard. 
\end{corollary}
One immediately gets that the isomorphism problem for tree automatic
scattered linear orders is $\Pi^0_1$-hard. We do not know whether the
set of tree automatic presentations of \emph{scattered} linear orders
is decidable. Therefore, the following immediate consequence of
Corollary~\ref{C-Iso-tree-automatic} is a bit stronger:
\begin{corollary}
  Let $X$ be a set of pairs of tree automatic presentations such that,
  for all tree automatic presentations $P_1$ and $P_2$ of scattered
  linear orders $\cL_1$ and $\cL_2$, one has
  \[
    (P_1,P_2)\in X\iff \cL_1\cong\cL_2\,.
  \]
  Then $X$ is $\Pi^0_1$-hard.
\end{corollary}

The rest of this section is devoted to the proof of
Lemma~\ref{L-weighted-automata}.

\begin{proof}[Proof of Lemma~\ref{L-weighted-automata}]
  Let $\cA=(Q,\Sigma,\iota,\mu,F)$ be a weighted automaton. We will
  construct a tree automatic presentation of the linear
  order~$\cL_{||\cA||}$.

  A \emph{run tree of $\cA$} is a tree $t$ over the alphabet
  $\Sigma\uplus\{\$\}$ such that there exist states
  $\iota=q_0,q_1,\dots,q_{k-1}\in Q$ and $q_k\in F$ (with
  $k=\max\{i\mid 0^{i+1}\in\dom(t)\}$) with the following properties:
  \begin{enumerate}[(T1)]
  \item $11\in\dom(t)\subseteq 0^*\cup 0^*10^*\cup 110^*$ and
    $100\notin\dom(t)$
  \item $t(0^i)\in\Sigma$ and $\mu(q_{i-1},t(0^i),q_i)\neq-\infty$ for
    all $1\le i\le k$
  \item $0^i1\in\dom(t)$ implies $1\le i\le k$ and
    $\mu(q_{i-1},a_i,q_i)=1$ or $i=0$
  \item $t^{-1}(\$)=\dom(t)\setminus \{0^i\mid 1\le i\le k\}$
  \end{enumerate}
  Note that every run tree $t$ defines a word over $\Sigma$, namely
  \[
    \word(t)=t(0)\,t(00)\,\dots\, t(0^k)\,.
  \]
  Since $11\in\dom(t)$, also $1$ and therefore $0$ belong to
  $\dom(t)$. Hence $\word(t)\neq\varepsilon$. Fig.~\ref{fig:run-tree}
  shows a run tree~$t$ with $\word(t)=abaab$ (we omitted the label
  $\$$ in the figure). The idea is that the ``main branch''
  $\{0,00,\dots,0^k\}$ carries a run~$\rho$ of the weighted
  automaton~$\cA$. The number of ``side branches'' starting in some
  node~$0^i1$ with $i>0$ is at most the weight~$\wt(\rho)$ of the
  encoded run. Since these side branches have arbitrary length, the
  whole run tree stands for an element of $\omega^{\wt(\rho)}$. The
  ``side branch'' starting in $11$ plays a special role, its length
  $|\dom(t)\cap 110^+|$ is denoted $n(t)$ (the run tree $t$ in
  Fig.~\ref{fig:run-tree} satisfies $n(t)=2$).

  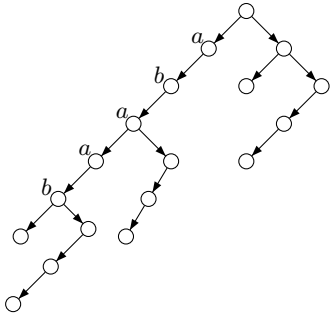
\begin{figure}
    \centering
    {
      \begin{picture}(50,50)(-20,5)
        \gasset{Nw=2,Nh=2,ExtNL=y,NLangle=135}
        \node(1)(15,50){} \node(2)(10,45){$a$} \node(3)(5,40){$b$}
        \node(4)(0,35){$a$} \node(5)(-5,30){$a$} \node(6)(-10,25){$b$} 
        \node(7)(-15,20){}
        \drawedge(1,2){} \drawedge(2,3){} \drawedge(3,4){} \drawedge(4,5){}
        \drawedge(5,6){} \drawedge(6,7){}

        \node(8)(20,45){} \node(9)(15,40){}
        \node(10)(25,40){} \node(11)(20,35){} \node(12)(15,30){}
        \drawedge(1,8){} \drawedge(8,9){} \drawedge(8,10){}
        \drawedge(10,11){} \drawedge(11,12){}

        \node(13)(5,30){} \node(14)(2,25){} \node(15)(-1,20){}
        \drawedge(4,13){} \drawedge(13,14){} \drawedge(14,15){}

        \node(16)(-6,21){} \node(17)(-11,16){} \node(18)(-16,11){} 
        \drawedge(6,16){} \drawedge(16,17){} \drawedge(17,18){}
      \end{picture}
    }
    \caption{A run tree}
    \label{fig:run-tree}
  \end{figure}

  We next define, for two trees $s$ and $t$, the tree $s+t$ by adding
  a new $\$$-labeled root and considering $s$ as left subtree of $s+t$
  and $t$ as right subtree. More formally,
  $\dom(s+t)=\{\varepsilon\}\cup 0\dom(s)\cup 1\dom(t)$,
  $(s+t)(\varepsilon)=\$$, $(s+t)(0u)=s(u)$ for $u\in\dom(s)$, and
  $(s+t)(1v)=t(v)$ for $v\in\dom(t)$. Since we consider words as
  special trees, we will meet trees of the form $w+t$. These trees
  carry the sequences $\$w$ on $\dom(w+t)\cap 0^*$ and satisfy
  $(w+t)\restriction_1\cong t$.

  We now define the language $L_\cA$ by
  \[
    L_\cA = \{t\mid t\text{ is a run tree}\}
            \cup \{w\$+t\mid w\in\Sigma^+,t\in D\}
  \]
  where $D$ is the set of trees from page~\pageref{page-definition-D}
  that satisfies $(D;\le_\trees)\cong\delta$. This language is clearly
  regular.

  Note that trees from $L_\cA$ use the alphabet $\Sigma\cup\{\$\}$
  that we order arbitrarily. We will now prove
  \[
     (L_\cA;\le_\trees)\cong\cL_{||\cA||}\,.
  \]

  First let $w\in\Sigma^+$ and $n\in\bN$. Then let $\cI^1_{w,n}$
  denote the restriction of $(L_\cA;\le_\trees)$ to all run trees $t$
  with
  \begin{equation}
    \word(t)=w
    \text{ and }n(t)=n\,.\label{eq:AAA}
  \end{equation}
  Note that for any two run trees $s$ and $t$ satisfying
  \eqref{eq:AAA}, we have $s\rest_{0^*}=t\rest_{0^*}$ and
  $s\rest_1=t\rest_1$. Hence $s\le_\trees t$ if and only if there
  exists $i\ge 1$ with $t\rest_{0^j1}=s\rest_{0^j1}$ for all $1\le
  j<i$ and $t\rest_{0^i1} <_\trees s\rest_{0^i1}$. By (T3),
  $\dom(t)\cap 0^+1$ contains at most $|w|$ elements. Furthermore note
  that the trees $t\rest_{0^i1}$ can be identified with natural
  numbers (namely with $|\dom(t)\cap 0^i10^*|$). This shows that
  $\cI^1_{w,n}$ can be embedded into $(\bN^{|w|};\le_\lex)$ and is
  therefore well-ordered and at most $\omega^{|w|}$.

  Now let $\rho=(q_0,a_1,q_1)(q_1,a_2,q_2)\dots(q_{k-1},a_{k},q_k)
  \in\Run(\cA,w)$ be a run of the weighted automaton~$\cA$ on the
  word~$w=a_1\dots a_k$. For any tuple~$(m_1,\dots,m_k)\in\bN^k$ such
  that
  \[
    m_i>0\ \Longrightarrow\ \mu(q_{i-1},a_{i},q_{i})=1\,,
  \]
  there exists a unique run tree~$t$ satisfying \eqref{eq:AAA} and
  $|\dom(t)\cap 0^i10^*|=m_i$ for all $1\le i\le k$. This gives an
  order-preserving embedding $f_\rho:\omega^{\wt(\rho)}\to
  \cI_{w,n}^1$, i.e., we showed
  $\omega^{\wt(\rho)}\le\cI_{w,n}^1$. Since this holds for abitrary
  runs $\rho\in\Run(\cA,w)$ and since $||\cA||(w)=\max\{\wt(\rho)\mid
  \rho\in\Run(\cA,w)\}$, we get $\omega^{||\cA||(w)}\le\cI_{w,n}^1$
  and therefore
  \[
  \omega^{||\cA||(w)+1}\le\cI_{w,n}^1\cdot\omega\,.
  \]
  By (T2), for every run tree $t$ satisfying \eqref{eq:AAA}, there
  exists at least one run $\rho\in\Run(\cA,w)$ such that $t$ is in the
  image of the embedding~$f_\rho$. Hence
  \[
  \cI_{w,n}^1\le\bigoplus_{\rho\in\Run(\cA,w)}\omega^{\wt(\rho)}
  \]
  where $\bigoplus$ denotes the natural sum of ordinals. We can
  conclude
  \begin{align*}
    \omega^{||\cA||(w)+1} \le \cI_{w,n}^1\cdot\omega
    &\le \left(\bigoplus_{\rho\in\Run(\cA,w)}\omega^{\wt(\rho)}\right)\cdot\omega\\
    &=\omega^{\max\{\wt(\rho)\mid\rho\in\Run(\cA,w)\}+1}\\
    &=\omega^{||\cA||(w)+1}
  \end{align*}
  and therefore
  \[
    \cI_{w,n}^1\cdot\omega = \omega^{||\cA||(w)+1}\,.
  \]

  Next consider the restriction $\cI_{w}^1$ of $(L_\cA;\le_\trees)$ to
  the set of run trees~$t$ with $\word(t)=w$. Then $n(s)<n(t)$ implies
  $s <_\trees t$. Furthermore, the restriction of $\cI_w^1$ to the set
  of run trees $t$ with $n(t)=n$ equals $\cI_{w,n}^1$. Hence
  \[
    \cI_w^1=\sum_{n\in(\bN;\le)}\cI_{w,n}^1 = 
    \cI_{w,0}^1\cdot\omega=\omega^{||\cA||(w)+1}\,.
  \]

  Next consider the restriction $\cI_w^2$ of $(L_\cA;\le_\trees^2)$ to
  the set of trees $w\$+D$. Then $\cI_w^2\cong\delta$ by what we saw
  on page~\pageref{page-definition-D}. Let $s$ be a run tree with
  $\word(s)=w$ and let $t\in w\$+D$. Then $s$ and $t$ coincide on
  $0^*$ (where they both carry the sequence $\$w\$$). Consider
  $s\rest_{10^*}$ and $t\rest_{10^*}$. Since $s$ is a run tree, we
  have $\dom(s)\cap 10^*=\{1,10\}$ while $t\rest_1\in D$ implies
  $\dom(t)\cap 10^*=\{1,10,100\}$. Hence $s\rest_1 <_\trees
  t\rest_1$ and therefore $s <_\trees t$. Hence, the restriction
  $\cI_w$ of $(L_\cA;\le_\trees)$ to the set of run trees $t$ with
  $\word(t)=w$ and the set of trees $w\$+ D$ satisfies
  \[
    \cI_w= \cI_w^1+\cI_w^2\cong\omega^{||\cA||(w)+1}+\delta\,.
  \]

  Finally, let $u,v\in\Sigma^+$. Then $u\le_\llex v$ if and only if
  $u\le_\trees v$. This implies
  \[
    (L_\cA;\le_\trees) = \sum_{w\in(\Sigma^+;\le_\llex)}\cI_w
      \cong\sum_{w\in(\Sigma^+;\le_\llex)}\omega^{||\cA||(w)+1}+\delta
      =\cL_\cA\,.
  \]
\end{proof}

\subsection{Automorphisms}

From Theorem~\ref{T-rigid-automatic}, we already know that the
existence of a non-trivial automorphism of a word automatic and
scattered linear order is $\Sigma^0_1$-hard. Here, we push this lower
bound one level higher for tree automatic scattered linear orders. The
order theoretic construction resembles that from
Section~\ref{S-word-rigid}, but also uses ideas from the previous
section.

Let $\cM$ be a Minsky machine, let $\cA_\cM$ and $\cB_\cM$ be the
weighted automata and, for $m\in\bN$, let $r_{\cM,m}$ be the function
defined following Theorem~\ref{T-Krob}. Then we define the linear
order
\[
  \cL_\cM=\sum_{m\in(\bN;\le)}
     \left(
       \cL_{||\cA_\cM||}\cdot\omega^*+
       \cL_{r_{\cM,m}}\cdot\omega
     \right)\,.
\]

\begin{lemma}\label{L-tree-automatic-2}
  From a Minsky machine $\cM$, one can construct a tree automatic
  presentation of the linear order~$\cL_\cM$.
\end{lemma}

\begin{proof}
  Let $\cM$ be a Minsky machine, let $\cA_\cM$ and $\cB_\cM$ be the
  weighted automata and let $r_{\cM,m}:\Sigma^+\to\bN$ be the function
  defined following Theorem~\ref{T-Krob}. Recall that the alphabet of
  $\cA_\cM$ is $\Sigma$ and that of $\cB_\cM$ is
  $\Sigma_{\#}^2$. Recall the notion of a run tree from the proof of
  Lemma~\ref{L-weighted-automata} that is based on a weighted
  automaton. In this proof, we will consider run trees with respect to
  the weighted automaton $\cA_\cM$ and with respect to the weighted
  automaton $\cB_\cM$. Now recall the definition of the language
  $L_{\cA_\cM}$ and $L_{\cB_\cM}$:
  \begin{eqnarray*}
    L_{\cA_\cM} &=& \{t\mid t\text{ is a run tree wrt.\ }\cA_\cM\}
            \cup \{w\$+t\mid w\in\Sigma^+,t\in D\}\\
    L_{\cB_\cM} &=& \{t\mid t\text{ is a run tree wrt.\ }\cB_\cM\}
            \cup \{w\$+t\mid w\in(\Sigma_{\#}^2)^+,t\in D\}
  \end{eqnarray*}
  Note that these two tree languages are disjoint since the alphabets
  $\Sigma$ and $\Sigma_{\#}^2$ are disjoint. Now define the language
  \begin{eqnarray*}
    L_\cM&=&(L_{\cA_\cM}\otimes\$^*\otimes\$\Box(a\Box)^*)\\
          &&\cup\,
          \{t\otimes\$^k\otimes\$\Box(a\Box)^m\mid
          \begin{array}[t]{l}
            m\in\bN, t\in L_{\cB_\cM},\text{ and }\\
            (t\text{ is a run tree }\Rightarrow
              \word(t)\in\Sigma^+\otimes\$\Box(a\Box)^m)\}\,.
          \end{array}
  \end{eqnarray*}
  The crucial point regarding the regularity of this set is the
  verification that a tree $s\otimes\$^k\otimes\$\Box(a\Box)^m$ with
  $t$ a run tree of $\cB_\cM$ belongs to the second set. But this is
  the case if $s\rest_{0^*}=\$\$\Box(a\Box)^m\$$, a property that a
  tree automaton can check easily.

  On this set, we define the following linear order $\preceq$:
  $(s\otimes\$^k\otimes \$\Box(a\Box)^m)\preceq
  (t\otimes\$^\ell\otimes \$\Box(a\Box)^n)$ if and only if one of the
  following hold
  \begin{enumerate}[(O1)]
  \item $m<n$ or
  \item $m=n$, $s\in L_{\cA_\cM}$, and $t\in L_{\cB_\cM}$, or
  \item $m=n$, $s,t\in L_{\cA_\cM}$, and $k>\ell$, or
  \item $m=n$, $s,t\in L_{\cA_\cM}$,  $k=\ell$, and $s\le_\trees t$, or
  \item $m=n$, $s,t\in L_{\cB_\cM}$, and $k<\ell$, or
  \item $m=n$, $s,t\in L_{\cB_\cM}$,  $k=\ell$, and $s\le_\trees t$.
  \end{enumerate}
  It is clear that this relation is automatic and it remains to be
  shown that $(L_\cM;\preceq)\cong\cL_\cM$.

  For $k,m\ge0$ let $\cI_{\cA,k,m}$ denote the restriction of
  $(L_\cM;\preceq)$ to the set $L_{\cA_\cM}\otimes\$^k\otimes
  \$\Box(a\Box)^m$. By (O4) and Lemma~\ref{L-weighted-automata}, we get
  \begin{equation}
    \label{eq:(I)}
    \cI_{\cA,k,m}\cong\cL_{||\cA_\cM||}\,.
  \end{equation}
  Next let $\cI_{\cA,m}$ denote the restriction of $(L_\cM;\preceq)$
  to the set $L_{\cA_\cM}\otimes\$^*\otimes \$\Box(a\Box)^m$. Then,
  (O3) and \eqref{eq:(I)} imply
  \begin{equation}
    \label{eq:(II)}
    \cI_{\cA,m}\cong\cL_{||\cA_\cM||}\cdot\omega^*\,.
  \end{equation}
 
  Now we consider the weighted automaton $\cB_\cM$: For $k,m\ge0$ let
  $\cI_{\cB,k,m}$ denote the restriction of $(L_\cM;\preceq)$ to the
  set of all trees $t\otimes\$^k\otimes \$\Box(a\Box)^m$ such that
  $t\in L_{\cB_\cM}$ and, if $t$ is a run tree of $\cB_\cM$, then
  $\word(t)\in\Sigma^+\otimes \$\Box(a\Box)^m$. By (O6),
  $\cI_{\cB,k,m}$ is a restriction of
  $(L_{\cB_\cM};\le_\trees)\cong\cL_{\cB_\cM}$. Using the arguments
  from the proof of Lemma~\ref{L-weighted-automata} again, we obtain
  \[
    \cI_{\cB,k,m}\cong\cL_{r_{\cM,m}}\,.
  \]
  Together with (O5), this implies
  \begin{equation}
    \label{eq:(IV)}
    \cI_{\cB,m}\cong\cL_{r_{\cM,m}}\cdot\omega
  \end{equation}
  where $\cI_{\cB,m}$ is the restriction of $(L_\cM;\preceq)$ to the
  set of all trees $t\otimes\$^k\otimes \$\Box(a\Box)^m$ such that
  $k\ge0$ is arbitrary, $t\in L_{\cB_\cM}$ and, if $t$ is a run tree
  of $\cB_\cM$, then $\word(t)\in\Sigma^+\otimes \$\Box(a\Box)^m$.

  Now, from (O2), \eqref{eq:(II)} and \eqref{eq:(IV)}, we obtain that
  the restriction of $(L_\cM;\preceq)$ to the set of trees that define
  $\cI_{\cA,m}$ and $\cI_{\cB,m}$ is isomorphic to
  \[
    \cL_{||\cA_\cM||}\cdot\omega^*+\cL_{r_{\cM,m}}\cdot\omega\,.
  \]
  Finally, (O1) implies
  \begin{eqnarray*}
     (L_\cM;\preceq)
       &\cong&\sum_{m\in(\bN;\le)}   
          \cL_{||\cA_\cM||}\cdot\omega^*+\cL_{r_{\cM,m}}\cdot\omega\\
       &=& \cL_\cM\,.
  \end{eqnarray*}
\end{proof}

\begin{theorem}\label{T-rigid-tree-automatic}
  \begin{enumerate}[(i)]
  \item The set of tree automatic presentations of rigid (rigid and
    scattered, resp.) linear orders is $\Pi^0_2$-hard.
  \item The set of tree automatic presentations of automatically rigid
    linear orders is $\Pi^0_1$-complete.
  \end{enumerate}
\end{theorem}

\begin{proof}
  \begin{enumerate}[(i)]
  \item As in Section~\ref{S-word-rigid}, the linear order $\cL_\cM$
    is rigid if and only if the block
    \[
    B(m)= \cL_{||\cA_\cM||}\cdot\omega^*+\cL_{r_{\cM,m}}\cdot\omega
    \]
    is rigid for every $m$. But the linear order $B(m)$ is rigid if
    and only if $\cL_{||\cA_\cM||}\not\cong\cL_{r_{\cM,m}}$ if and
    only if $m$ is accepted by $\cM$ (by Theorem~\ref{T-Krob}). In
    summary, $\cL_\cM$ is rigid if and only if the Minsky
    machine~$\cM$ accepts all numbers, a $\Pi^0_2$-complete
    problem. Hence Lemma~\ref{L-tree-automatic-2} reduces a
    $\Pi^0_2$-complete problem to the set of tree-automatic
    presentations of rigid (rigid and scattered, resp.)  linear
    orders.
  \item Hardness follows from Corollary~\ref{C-rigid-automatic}(ii),
    containment in $\Pi^0_1$ can be shown as in the proof of
    Corollary~\ref{C-rigid-automatic}(ii).
  \end{enumerate}
\end{proof}

\section{Open questions}

The isomorphism and rigidity problems for word automatic scattered
linear orders both belong to $\Delta^0_\omega$ (cf.~\cite{KusLL11}),
our lower bound $\Pi^0_1$ for the rigidity problem leaves quite some
room for improvements. Since the rank of a tree automatic linear order
is properly below $\omega^\omega$~\cite{JaiKSS12,Hus12}, the proof of
\cite{KusLL11} can be adapted to show that the isomorphism and the
rigidity problems for tree automatic scattered linear orders both
belong to $\Sigma^0_{\omega^\omega}$. But we only have the lower
bounds $\Pi^0_1$ and $\Pi^0_2$, resp. Finally, the rigidity problem
for arbitrary word or tree automatic linear orders is in $\Pi^1_1$,
but also here, we only have the arithmetic lower bound $\Pi^0_1$ and
$\Pi^0_2$, resp.

But the most pressing open question is the isomorphism problem of
scattered and word automatic linear orders.


\begin{thebibliography}{10}

\bibitem{BarGR11}
V.~B{\'a}r{\'a}ny, E.~Gr{\"a}del, and S.~Rubin.
\newblock Automata-based presentations of infinite structures.
\newblock In {\em Finite and Algorithmic Model Theory}, pages 1--76. Cambridge
  University Press, 2011.

\bibitem{BloE05}
S.L. Bloom and Z.~{\'E}sik.
\newblock The equational theory of regular words.
\newblock {\em Information and Computation}, 197:55--89, 2005.

\bibitem{BluG00}
A.~Blumensath and E.~Gr{\"a}del.
\newblock Automatic {S}tructures.
\newblock In {\em LICS'00}, pages 51--62. IEEE Computer Society Press, 2000.

\bibitem{BluG04}
A.~Blumensath and E.~Gr{\"a}del.
\newblock Finite presentations of infinite structures: Automata and
  interpretations.
\newblock {\em Theory of Computing Systems}, 37(6):641--674, 2004.

\bibitem{BraS11}
G.~Braun and L.~Str\"ungmann.
\newblock Breaking up finite automata presentable torsion-free abelian groups.
\newblock {\em International Journal of Algebra and Computation},
  21(8):1463--1472, 2011.

\bibitem{CarE12}
A.~Carayol and Z.~{\'E}sik.
\newblock The {FC}-rank of a context-free language.
\newblock arXiv:1202.6275, February 2012.

\bibitem{Cou78}
B.~Courcelle.
\newblock Frontiers of infinite trees.
\newblock {\em RAIRO - Theoretical Informatics}, 12(4):319--337, 1978.

\bibitem{DroKV09}
M.~Droste, W.~Kuich, and H.~Vogler, editors.
\newblock {\em Handbook of Weighted Automata}.
\newblock EATCS Monographs in Theoretical Computer Science. Springer, 2009.

\bibitem{DroK}
M.~Droste and D.~Kuske.
\newblock Weighted automata.
\newblock To appear in the forthcoming handbook AutoMathA, 2012.

\bibitem{DurH12}
A.~Durand-Gasselin and P.~Habermehl.
\newblock Ehrenfeucht-{F}ra\"\i{}ss\'e goes elementarily automatic for
  structures of bounded degree.
\newblock In {\em STACS'12}, pages 242--253. Dagstuhl Publishing, 2012.

\bibitem{Elg61}
C.C. Elgot.
\newblock Decision problems of finite automata design and related arithmetics.
\newblock {\em Trans. Am. Math. Soc.}, 98:21--51, 1961.

\bibitem{EpsCHLPT92}
D.B.A. Epstein, J.W. Cannon, D.F. Holt, S.V.F. Levy, M.S. Paterson, and W.P.
  Thurston.
\newblock {\em Word Processing In Groups}.
\newblock Jones and Bartlett Publishers, Boston, 1992.

\bibitem{Esi11}
Z.~{\'E}sik.
\newblock An undecidable property of context-free linear orders.
\newblock {\em Inform.\ Processing Letters}, 111(3):107--109, 2011.

\bibitem{Hei80}
St. Heilbrunner.
\newblock An algorithm for the solution of fixed-point equations for infinite
  words.
\newblock {\em RAIRO -- Theoretical Informatics}, 14(2):131--141, 1980.

\bibitem{Hod82}
B.R. Hodgson.
\newblock On direct products of automaton decidable theories.
\newblock {\em Theoretical Computer Science}, 19:331--335, 1982.

\bibitem{Hus12}
M.~Huschenbett.
\newblock The rank of tree-automatic linear orderings.
\newblock \url{http://arxiv.org/abs/1204.3048}, 2012.

\bibitem{JaiKSS12}
S.~Jain, B.~Khoussainov, Ph. Schlicht, and F.~Stephan.
\newblock Tree-automatic scattered linear orders.
\newblock Manuscript, 2012.

\bibitem{KhoN95}
B.~Khoussainov and A.~Nerode.
\newblock Automatic presentations of structures.
\newblock In {\em Logic and Computational Complexity}, Lecture Notes in Comp.\
  Science vol.\ 960, pages 367--392. Springer, 1995.

\bibitem{KhoN08}
B.~Khoussainov and A.~Nerode.
\newblock Open questions in the theory of automatic structures.
\newblock {\em Bulletin of the EATCS}, 94:181--204, 2008.

\bibitem{KhoNRS07}
B.~Khoussainov, A.~Nies, S.~Rubin, and F.~Stephan.
\newblock Automatic structures: richness and limitations.
\newblock {\em Log. Methods in Comput. Sci.}, 3(2), 2007.

\bibitem{KhoRS05}
B.~Khoussainov, S.~Rubin, and F.~Stephan.
\newblock Automatic linear orders and trees.
\newblock {\em ACM Transactions on Computational Logic}, 6(4):675--700, 2005.

\bibitem{Kro94}
D.~Krob.
\newblock The equality problem for rational series with multiplicities in the
  tropical semiring is undecidable.
\newblock {\em International Journal of Algebra and Computation},
  4(3):405--425, 1994.

\bibitem{KusLL11}
D.~Kuske, J.~Liu, and M.~Lohrey.
\newblock The isomorphism problem on classes of automatic structures with
  transitive relations.
\newblock {\em Transactions of the {AMS}}, 2011.
\newblock Accepted.

\bibitem{LohM11}
M.~Lohrey and Ch. Mathissen.
\newblock Isomorphism of regular trees and words.
\newblock In {\em ICALP'11}, Lecture Notes in Comp.\ Science vol.\ 6756, pages
  210--221. Springer, 2011.

\bibitem{Mat93}
Y.~Matijasevich.
\newblock {\em Hilbert's Tenth Problem}.
\newblock Foundations of Computing Series. MIT Press, 1993.

\bibitem{Nie07}
A.~Nies.
\newblock Describing groups.
\newblock {\em Bulletin of Symbolic Logic}, 13(3):305--339, 2007.

\bibitem{Rub04}
S.~Rubin.
\newblock {\em Automatic {S}tructures}.
\newblock PhD thesis, University of Auckland, 2004.

\bibitem{Rub08}
S.~Rubin.
\newblock Automata presenting structures: A survey of the finite string case.
\newblock {\em Bulletin of Symbolic Logic}, 14:169--209, 2008.

\bibitem{Tho86}
W.~Thomas.
\newblock On frontiers of regular trees.
\newblock {\em RAIRO -- Theoretical Informatics}, 20(4):371--381, 1986.

\end{thebibliography}
\end{document}